\documentclass[11pt]{article}
\usepackage[margin=1in]{geometry}

\let\coloneqq\relax

\usepackage{amsfonts}
\usepackage{amsthm}
\usepackage{amssymb}
\usepackage{amsmath}
\usepackage{bbold}
\usepackage{cite}
\usepackage{bbm}
\usepackage{nonfloat}
\usepackage[pdftex]{hyperref}
\usepackage{braket}
\usepackage{dsfont}
\usepackage{mathdots}
\usepackage{mathtools}
\usepackage{enumerate}
\usepackage[shortlabels]{enumitem}
\usepackage{stmaryrd}
\usepackage{graphicx}
\usepackage{stackengine}
\usepackage{scalerel}
\usepackage{xr}
\usepackage[dvipsnames]{xcolor}
\usepackage{xcolor}
\usepackage{array}
\usepackage{makecell}
\newcolumntype{x}[1]{>{\centering\arraybackslash}p{#1}}
\usepackage{tikz}
\usepackage{pgfplots}
\usetikzlibrary{shapes.geometric, shapes.misc, positioning, arrows, arrows.meta, decorations.pathreplacing, decorations.pathmorphing, patterns, angles, quotes, calc}
\usepackage{booktabs}
\usepackage{xfrac}
\usepackage{siunitx}
\usepackage{centernot}
\usepackage{comment}
\usepackage{chngcntr}
\usepackage[ruled,vlined]{algorithm2e}
\usepackage[capitalize,nameinlink]{cleveref}

\makeatletter
\def\thmhead@plain#1#2#3{%
  \thmname{#1}\thmnumber{\@ifnotempty{#1}{ }\@upn{#2}}%
  \thmnote{ {\the\thm@notefont#3}}}
\let\thmhead\thmhead@plain
\makeatother

\newcommand{\bb}{\begin{equation}\begin{aligned}\hspace{0pt}}
\newcommand{\bbb}{\begin{equation*}\begin{aligned}}
\newcommand{\ee}{\end{aligned}\end{equation}}
\newcommand{\eee}{\end{aligned}\end{equation*}}
\newcommand*{\coloneqq}{\mathrel{\vcenter{\baselineskip0.5ex \lineskiplimit0pt \hbox{\scriptsize.}\hbox{\scriptsize.}}} =}

\newcommand{\Tr}{\mathrm{Tr}}

\DeclareMathAlphabet{\pazocal}{OMS}{zplm}{m}{n}

\newcommand{\lsmatrix}{\left(\begin{smallmatrix}}
\newcommand{\rsmatrix}{\end{smallmatrix}\right)}

\stackMath

\stackMath

\makeatletter
\newcommand*\rel@kern[1]{\kern#1\dimexpr\macc@kerna}
\newcommand*\widebar[1]{%
  \begingroup
  \def\mathaccent##1##2{%
    \rel@kern{0.8}%
    \overline{\rel@kern{-0.8}\macc@nucleus\rel@kern{0.2}}%
    \rel@kern{-0.2}%
  }%
  \macc@depth\@ne
  \let\math@bgroup\@empty \let\math@egroup\macc@set@skewchar
  \mathsurround\z@ \frozen@everymath{\mathgroup\macc@group\relax}%
  \macc@set@skewchar\relax
  \let\mathaccentV\macc@nested@a
  \macc@nested@a\relax111{#1}%
  \endgroup
}

\counterwithin*{equation}{part}
\counterwithin*{figure}{part}

\tikzset{meter/.append style={draw, inner sep=10, rectangle, font=\vphantom{A}, minimum width=30, line width=.8, path picture={\draw[black] ([shift={(.1,.3)}]path picture bounding box.south west) to[bend left=50] ([shift={(-.1,.3)}]path picture bounding box.south east);\draw[black,-latex] ([shift={(0,.1)}]path picture bounding box.south) -- ([shift={(.3,-.1)}]path picture bounding box.north);}}}
\tikzset{roundnode/.append style={circle, draw=black, fill=gray!20, thick, minimum size=10mm}}
\tikzset{squarenode/.style={rectangle, draw=black, fill=none, thick, minimum size=10mm}}

\definecolor{Blues5seq1}{RGB}{239,243,255}
\definecolor{Blues5seq2}{RGB}{189,215,231}
\definecolor{Blues5seq3}{RGB}{107,174,214}
\definecolor{Blues5seq4}{RGB}{49,130,189}
\definecolor{Blues5seq5}{RGB}{8,81,156}

\definecolor{Greens5seq1}{RGB}{237,248,233}
\definecolor{Greens5seq2}{RGB}{186,228,179}
\definecolor{Greens5seq3}{RGB}{116,196,118}
\definecolor{Greens5seq4}{RGB}{49,163,84}
\definecolor{Greens5seq5}{RGB}{0,109,44}

\definecolor{Reds5seq1}{RGB}{254,229,217}
\definecolor{Reds5seq2}{RGB}{252,174,145}
\definecolor{Reds5seq3}{RGB}{251,106,74}
\definecolor{Reds5seq4}{RGB}{222,45,38}
\definecolor{Reds5seq5}{RGB}{165,15,21}

\pgfplotsset{width=10cm,compat=1.9}
\usetikzlibrary{decorations.markings}

\newcommand{\RN}[1]{\textup{\uppercase\expandafter{\romannumeral#1}}}

\newcommand{\symspace}[2]{\mathrm{Sym}_{k}(\bb{C}^d)}
\newcommand{\symproj}[2]{P_{\mathrm{Sym}}^{(d,k)}}

\newtheorem{proposition}{Proposition}
\newtheorem{theorem}{Theorem}

\newtheorem{lemma}{Lemma}
\newtheorem{corollary}{Corollary}

\newtheorem{problem}{Problem}
\newtheorem{definition}{Definition}

\numberwithin{lemma}{section}
\numberwithin{proposition}{section}
\numberwithin{corollary}{section}
\numberwithin{definition}{section}
\numberwithin{theorem}{section}
\numberwithin{remark}{section}
\numberwithin{equation}{section}
\numberwithin{fact}{section}

\hypersetup{
	colorlinks = true,
	linkcolor = red!50!black,
	citecolor = green!40!black,
	urlcolor = blue!60!black}

\DeclareRobustCommand{\orcidicon}{%
	\begin{tikzpicture}
	\draw[lime, fill=lime] (0,0) 
	circle [radius=0.16] 
	node[white] {{\fontfamily{qag}\selectfont \tiny ID}};
	\draw[white, fill=white] (-0.0625,0.095) 
	circle [radius=0.007];
	\end{tikzpicture}
	\hspace{-2mm}
}
\foreach \x in {A, ..., Z}{%
	\expandafter\xdef\csname orcid\x\endcsname{\noexpand\href{https://orcid.org/\csname orcidauthor\x\endcsname}{\noexpand\orcidicon}}
}

\title{Optimal certification of constant-local Hamiltonians}

\author{
Junseo Lee$^{\,*,\,\dagger}$ \\
\href{mailto:harris.junseo@gmail.com}{\texttt{harris.junseo@gmail.com}}
\and
Myeongjin Shin$^{\,*,\,\ddagger}$ \\
\href{mailto:hanwoolmj@kaist.ac.kr}{\texttt{hanwoolmj@kaist.ac.kr}}
}

\footnotetext[0]{\textit{Authors contributed equally and are listed alphabetically by last name.}}
\footnotetext[1]{Team QST, Seoul National University, Seoul 08826, Korea}
\footnotetext[2]{Quantum AI Team, Norma Inc., Seoul 04799, Korea}
\footnotetext[3]{School of Computational Sciences, KAIST, Seoul 02455, Korea}

\date{\today}

\begin{document}
\maketitle

\begin{abstract}
We study the problem of certifying local Hamiltonians from real-time access to their dynamics. Given oracle access to $e^{-itH}$ for an unknown $k$-local Hamiltonian $H$ and a fully specified target Hamiltonian $H_0$, the goal is to decide whether $H$ is exactly equal to $H_0$ or differs from $H_0$ by at least $\varepsilon$ in normalized Frobenius norm, while minimizing the total evolution time. We introduce the first \emph{intolerant} Hamiltonian certification protocol that achieves optimal performance for all constant-locality Hamiltonians. For general $n$-qubit, $k$-local, traceless Hamiltonians, our procedure uses $\mathcal{O}(c^k/\varepsilon)$ total evolution time for a universal constant $c$, and succeeds with high probability. In particular, for $\mathcal{O}(1)$-local Hamiltonians, the total evolution time becomes $\Theta(1/\varepsilon)$, matching the known $\Omega(1/\varepsilon)$ lower bounds and achieving the gold-standard Heisenberg-limit scaling. Prior certification methods either relied on implementing inverse evolution of $H$, required controlled access to $e^{-itH}$, or achieved near-optimal guarantees only in restricted settings such as the Ising case ($k=2$). In contrast, our algorithm requires neither inverse evolution nor controlled operations: it uses only forward real-time dynamics and achieves optimal intolerant certification for \emph{all} constant-locality Hamiltonians.
\end{abstract}

\maketitle

\newpage
\tableofcontents

\newpage
\section{Introduction}

Time evolution generated by a self-adjoint Hamiltonian is a foundational structure in quantum mechanics.  
It underlies quantum simulation~\cite{feynman1982simulating,trabesinger2012quantum,georgescu2014quantum}, which aims to reproduce the dynamics $e^{-itH}$ of a target Hamiltonian $H$ on a controllable quantum device, and it plays a central role in applications across quantum many-body physics~\cite{jordan2012quantum,schreiber2015observation,eisert2015quantum,smith2016many,bauer2023quantum}.  
As quantum hardware continues to improve, it has become increasingly critical to understand, control, and certify the dynamics that a device implements in practice.

\medskip
Recent progress in quantum learning theory~\cite{arunachalam2017guest} has been driven by the broader goal of understanding how efficiently one can extract information from quantum systems.  
Depending on the task, the unknown object may be a quantum state, a quantum circuit, a dynamical process, or a Hamiltonian, and the learning objective differs accordingly.  
In some settings the aim is to recover the entire object through quantum tomography~\cite{anshu2024survey}; in others, the goal is to determine only a specific feature or structural property, as studied in quantum property testing~\cite{montanaro2013survey}; and in yet other scenarios, the task is to verify whether the object is consistent with a prescribed specification, which is the purpose of certification~\cite{eisert2020quantum,kliesch2021theory}.

\medskip
Within this broader landscape, a substantial line of work investigates how to \emph{learn} an unknown Hamiltonian from access to its real-time evolution.  
In this learning-from-dynamics paradigm, one can query the evolution operator $e^{-itH}$ for chosen times $t$, prepare suitable input states, let them evolve under $H$, and perform measurements to reconstruct a classical description of the underlying Hamiltonian. Such Hamiltonian learning algorithms can reveal detailed information about interaction patterns and coupling strengths, and have been developed across a broad spectrum of settings.  
Early proposals of Hamiltonian learning~\cite{da2011practical,shabani2011efficient,granade2012robust,wiebe2014quantum,wiebe2014hamiltonian} were followed by a large body of work developing numerous algorithmic improvements and addressing diverse problem scenarios, including a wide range of input models and resource constraints~\cite{qi2019determining,bairey2019learning,li2020hamiltonian,anshu2021sample,haah2022optimal,yu2023robust,huang2023learning,stilck2024efficient,gu2024practical,caro2024learning,bakshi2024learning,bakshi2024structure,dutkiewicz2024advantage,ma2024learning,rouze2024learning,castaneda2025hamiltonian,zhao2025learning,hu2025ansatz,sinha2025improved,abbas2025nearly,arunachalam2025testing,bluhm2025certifying,chen2025quantum}.  
Hamiltonian learning has also been extended beyond qubit systems to bosonic and fermionic models~\cite{hangleiter2024robustly,li2024heisenberg,fanizza2024efficient,ni2024quantum,mirani2024learning}.  
For a broader overview of the field, we refer the reader to the survey sections in~{\cite[Section~1.3]{gao2025quantum}}.

\medskip
A related line of work investigates \emph{Hamiltonian property testing}~\cite{bluhm2024hamiltonian}, where the goal is not to reconstruct $H$ itself but to decide whether it satisfies a given structural property or is far from every Hamiltonian with that property.  
Examples include locality testing~\cite{gutierrez2024simple,bluhm2024hamiltonian,arunachalam2025testing}, which determines whether $H$ acts nontrivially only on small subsets of qubits, and sparsity testing~\cite{sinha2025improved,arunachalam2025testing}, which checks whether $H$ has only a few nonzero interaction terms.  
In many experimental settings, one does not seek a full reconstruction of the Hamiltonian.  
Instead, the starting point is often a \emph{target} Hamiltonian $H_0$ that models the intended behavior of the device, and the key question is whether the implemented dynamics are consistent with this specification.  
This leads to the task of \emph{Hamiltonian certification}~\cite{gao2025quantum}, whose aim is to determine whether an unknown Hamiltonian $H$ matches or is close to $H_0$, or whether it deviates from it by a significant amount.  
In this framework, the natural resource measure is the \emph{total evolution time}, defined as the sum of all intervals during which the device is allowed to evolve under $H$, following conventions established in Hamiltonian learning. Focusing on local Hamiltonians, we ask:

\begin{center}
\emph{Given two local Hamiltonians, can one develop a procedure that, using only real-time access to their dynamics, efficiently certifies whether they are identical or $\varepsilon$-far?}
\end{center}

In this work, we study a formulation we refer to as \emph{intolerant Hamiltonian certification}.  
The goal is to distinguish between the exact equality $H = H_0$ and the alternative that $H$ differs from $H_0$ by at least a threshold $\varepsilon$ in normalized Frobenius norm (see~\cref{para:notation} for norm definitions), without assuming the existence of any intermediate gap region.  
We assume black-box access to the time-evolution operator $e^{-itH}$ and treat the total evolution time as the primary resource. 

\medskip
Our broader aim is to understand the fundamental total-evolution-time complexity of certifying $k$-local Hamiltonians for general~$k$.  
As a first step toward this goal, we obtain a \textit{tight} characterization of the task for the regime of $\mathcal{O}(1)$-local Hamiltonians:  
we show that intolerant certification in this setting admits an optimal protocol with total evolution time $\Theta(1/\varepsilon)$, thereby resolving the problem for constant locality.  
In our model, the algorithm is granted access only to \emph{forward} real-time evolution and is not allowed to implement inverse evolution or any controlled versions of $e^{-itH}$.  
We formalize the task as follows.

\begin{problem}[(Intolerant $k$-local Hamiltonian certification)]\label{problem:intolerant}
Let $H$ and $H_0$ be $n$-qubit, $k$-local, traceless Hamiltonians, and let $\|\cdot\|_F$ denote the normalized Frobenius norm.
Given parameters $\varepsilon>0$ and $\delta\in(0,1)$, the algorithm may query the forward real-time evolution operator $e^{-itH}$ for times $t$ of its choice, but is not allowed to implement inverse evolution or controlled access to $e^{-itH}$.
It must output \textsf{ACCEPT} if $H=H_0$ and \textsf{REJECT} if $\|H-H_0\|_F\ge \varepsilon$, with success probability at least $1-\delta$.
The objective is to minimize the total evolution time (the sum of all queried evolution intervals) and the total number of oracle queries.
\end{problem}

\subsection{Main results}

Our main theorem provides the first algorithm that achieves optimal evolution-time scaling for the intolerant certification task formalized in~\cref{problem:intolerant}.  
The guarantee holds for general $k$-local Hamiltonians and does not rely on any additional structural assumptions.

\begin{theorem}[(Informal statement of the main result, see~\cref{thm:main-algorithm} for details)]
For any $n$-qubit, $k$-local, traceless Hamiltonians $H$ and $H_0$, where $H_0$ is known exactly, there exists a certification algorithm for~\cref{problem:intolerant} with access to $e^{-itH}$ that uses $\mathcal{O}(c^k/\varepsilon)$ total evolution time for a universal constant $c$, and succeeds with high probability.
\end{theorem}

Specializing to constant locality yields evolution-time optimal certification.

\begin{corollary}[(Optimal intolerant certification for $\mathcal{O}(1)$-local Hamiltonians)]
For any $n$-qubit, $\mathcal{O}(1)$-local, traceless Hamiltonians $H$ and $H_0$, there exists an algorithm solving~\cref{problem:intolerant} using total evolution time $\mathcal{O}(1/\varepsilon)$, with high probability.
\end{corollary}

A matching lower bound of $\Omega(1/\varepsilon)$ is known. In particular, by the lower-bound construction of~\cite{kallaugher2025hamiltonianlocalitytestingtrotterized}, there exist constant-local Hamiltonians (e.g., the simple pair $H=\varepsilon X$ and $H_0=-\varepsilon X$) that require total evolution time at least $\Omega(1/\varepsilon)$ to distinguish in~\cref{problem:intolerant}. Therefore, our corollary achieves the optimal $\Theta(1/\varepsilon)$ scaling for constant-local Hamiltonians, attaining the gold-standard Heisenberg-limit precision in total evolution time. For the full certification procedure and explicit constants, see~\cref{alg:main-algorithm}.

\subsection{Technical overview}\label{sec:tech-overview}

We outline the main ideas behind our Hamiltonian certification procedure and how the analysis achieves total evolution time $\Theta(1/\varepsilon)$ for constant locality $k$.

\paragraph{Bell sampling as a spectral probe.}
For an $n$-qubit Hamiltonian $H$ with eigenvalues $\{\lambda_j\}_{j=1}^{2^n}$, Bell sampling applied
to $U(t)=e^{-itH}$ yields an identity-outcome probability
\begin{equation}\label{eq:overview-I}
    I(t)=\frac{1}{4^n}\sum_{j,k}\cos((\lambda_j-\lambda_k)t),
\end{equation}
so $I(t)$ depends only on eigenvalue \emph{differences}. In particular, if $H=0$ then $I(t)=1$ for
all $t$.

\paragraph{A gap statistic and an optimal-time sufficient condition.}
Define the proportion of $\varepsilon$-separated eigenvalue pairs
\[
\Lambda(X,\varepsilon) \coloneqq \frac{1}{N^2}\sum_{j,k:\,|\lambda_j-\lambda_k|\ge \varepsilon} 1,
\qquad N=2^n.
\]
If $\Lambda(H,\varepsilon)\ge d$, then a constant fraction of the terms in~\cref{eq:overview-I} oscillate at frequency at least $\varepsilon$, forcing $I(t)$ to drop below $1$ for some $t\in[0,2/\varepsilon]$. \cref{lem:spectral-condition} formalizes this: with
$m=\Theta(\log(1/\delta))$ samples $t\sim\mathrm{Unif}[0,2/\varepsilon]$, one finds a time satisfying
$I(t)\le 1-\Omega(d)$ with probability at least $1-\delta$, using total evolution time
$\mathcal{O}(\log(1/\delta)/\varepsilon)$.

\paragraph{Diagonal $k$-local Hamiltonians via Boolean analysis.}
When $H$ is diagonal in a Pauli-$Z$ product basis, its eigenvalues form a real function $f$ on the
Boolean cube with Fourier degree at most $k$. For independent $s,t\in\{0,1\}^n$, the gap
$F(s,t)=f(s)-f(t)$ is a degree-$k$ function on $2n$ Boolean variables. Hypercontractivity bounds the
fourth moment by $\mathds{E}[F^4]\le 9^k(\mathds{E}[F^2])^2$, and Paley--Zygmund yields
\begin{equation}
    \Pr\bigl[|F(s,t)|\ge \|H\|_F\bigr]\ge \frac14\cdot 9^{-k},    
\end{equation}
which is exactly $\Lambda(H,\|H\|_F)\ge 9^{-k}/4$ (\cref{pro:Z-basis}). Combined with \cref{lem:spectral-condition}, this gives $\mathcal{O}(9^k/\varepsilon)$-time certification for diagonal Hamiltonians.

\paragraph{Random reduction from general $k$-local to nearly diagonal.}
For a general $k$-local Hamiltonian, we pick a random ``diagonal'' Pauli subspace
$S=\bigotimes_{i=1}^n\{I,Q^{(i)}\}$ with $Q^{(i)}\in\{X,Y,Z\}$ i.i.d.\ uniform, and define
$H_{\mathrm{eff}}$ by keeping only Pauli terms in $S$. Each Pauli string of weight at most $k$
survives with probability $3^{-|P|}$, so $\mathds{E}[\|H_{\mathrm{eff}}\|_F^2]\ge 3^{-k}\|H\|_F^2$.
A second-moment bound and Paley--Zygmund inequality give the constant-probability guarantee
\[
\Pr\left[\|H_{\mathrm{eff}}\|_F \ge \frac{\|H\|_F}{\sqrt2 \cdot 3^{k/2}}\right]\ge \frac{1}{4\cdot 3^k}
\qquad\text{(\cref{pro:random-basis}).}
\]

\paragraph{Random twirling and stability of gap statistics.}
Writing $H-H_0=H_{\mathrm{eff}}+H_1'$, we apply the twirling map
$X\mapsto \frac12(X+PXP)$ with $P\sim\mathrm{Unif}(S)$ for $T=\Theta(k)$ steps.
It fixes $H_{\mathrm{eff}}$ and contracts the orthogonal component in Frobenius norm, yielding
$\|H_T'\|_F\le 2\cdot 2^{-T/2}\|H-H_0\|_F$ with high probability
(\cref{pro:z-twirl}). Hoffman--Wielandt inequality (\cref{pro:eig-stability}) then shows that the gap statistic $\Lambda(\cdot,\cdot)$ degrades only by $\mathcal{O}(\|H_T'\|_F^2)$, so $H_T$ still inherits a $\Lambda$-lower bound comparable to that of $H_{\mathrm{eff}}$.

\paragraph{Putting everything together.}
Conditioned on the above events, we obtain a Hamiltonian $H_T$ satisfying
$\Lambda(H_T,\eta)\ge \Theta(9^{-k})$ for $\eta=\Theta(\|H_{\mathrm{eff}}\|_F)\ge
\Theta(\varepsilon/3^{k/2})$. \cref{lem:spectral-condition} then guarantees a random time $t\in[0,2/\eta]$ (hence $t\le b=\Theta(3^{k/2}/\varepsilon)$) with $I_T(t)\le 1-\Theta(9^{-k})$ with constant probability. A one-sided Bell-sampling test with $m=\Theta(9^k)$ shots detects this drop with constant power, and repeating $R=\Theta(3^k\log(1/\delta))$ rounds amplifies success to $1-\delta$. Since each shot uses evolution time $\mathcal{O}(t)$ (via Trotterization), the total
evolution time is $\mathcal{O}({3^k\cdot 9^k\cdot 3^{k/2}\log({1}/{\delta})}/{\varepsilon})$ as claimed in~\cref{thm:main-algorithm}.

\subsection{Related work and our contributions}
Here we briefly summarize the prior results most relevant to our setting and highlight how our work differs from, sharpens, or extends these approaches.

\paragraph{Certification of Ising Hamiltonians.}
Bluhm et al.~\cite{bluhm2025certifying} studied the certification and learning of quantum Ising Hamiltonians.  
They showed that, for 2-local Ising Hamiltonians, certification in normalized Frobenius norm using access to $e^{-itH}$ can be achieved with total evolution time $\tilde{\mathcal{O}}(1/\varepsilon)$, matching the Heisenberg-scaling lower bound $\Omega(1/\varepsilon)$ up to logarithmic factors; a key ingredient in their analysis was the Bonami lemma from Boolean Fourier analysis.  
Our work extends and sharpens this picture: in the special case $k=2$ we obtain a fully optimal $\Theta(1/\varepsilon)$ certification algorithm (removing the logarithmic overhead), and moreover we achieve the same optimal scaling for all $\mathcal{O}(1)$-local Hamiltonians, not only Ising-type interactions.

\paragraph{Certification with inverse and controlled evolution.}
The setting most closely related to ours is the Hamiltonian certification framework introduced by Gao et al.~\cite{gao2025quantum}.  
They formalized the task of deciding whether an unknown Hamiltonian $H$ is $\varepsilon_1$-close to or $\varepsilon_2$-far from a target Hamiltonian $H_0$ (under normalized Frobenius and related norms), and gave direct certification protocols that achieved optimal total evolution time $\Theta((\varepsilon_2 - \varepsilon_1)^{-1})$, together with matching lower bounds.  
Their results applied to general Hamiltonians without assuming any locality or structural constraints, and covered both intolerant and tolerant formulations, Pauli norms and normalized Schatten $p$-norms for $1 \le p \le 2$, as well as ancilla-free schemes.  
A key feature of their model, however, is that it allows access not only to inverse evolution of the target Hamiltonian $H$ but also to controlled implementations of its dynamics.

In contrast, we focus on the intolerant setting ($\varepsilon_1 = 0$) and restrict attention to $k$-local Hamiltonians, and we show that for $\mathcal{O}(1)$-local Hamiltonians there exists a certification algorithm that achieves optimal $\Theta(1/\varepsilon)$ total evolution time \emph{without} using inverse evolution or any controlled operations.  
This provides a partial answer to the question posed in~\cite{gao2025quantum} of whether certification becomes intrinsically harder when inverse or controlled evolution of the target is unavailable, by showing that in the constant-locality regime the optimal scaling in total evolution time can still be attained.

\paragraph{Improved Bell sampling bounds.}
Sinha and Tong~\cite{sinha2025improved} established sharp short-time bounds on Bell sampling.  
For any traceless Hamiltonian $H$ with operator norm $\|H\|_{\mathrm{op}} \le L$, the probability $\Pr[I]$ of obtaining the Bell outcome $I$ after time evolution $e^{-iHt}$ satisfies
\begin{equation}
    1 - t^2 \|H\|_F^{\,2}  \le \Pr[I] \le 1 - 2ct^2\|H\|_F^{\,2},
\end{equation}
for any $c \in (0,1/2)$, provided that $t \le {t^*(c)}/{(2L)}$, where $t^*(c) \in (0,2\pi)$ is defined implicitly by $\cos(t^*(c)) = 1 - c(t^*(c))^2$. These inequalities capture the quadratic short-time decay of $\Pr[I]$, but the admissible time window  
$t \le t^*(c)/(2L)$ is governed by the operator norm of $H$, and therefore remains in the regime  
\(t = \mathcal{O}(1/\|H\|_{\mathrm{op}})\).  
This is too restrictive for certification tasks requiring evolution times on the order of \(1/\varepsilon\).

A natural question is whether similar two-sided control of $\Pr[I]$ can be extended to the longer-time regime $t = \Theta(1/\varepsilon)$, under the promise $\|H\|_F \ge \varepsilon$. If such bounds were available, then one could directly distinguish the cases  
\(\|H\|_F = 0\) and \(\|H\|_F \ge \varepsilon\) using only $\Theta(1/\varepsilon)$ total evolution time. We show that this long-time extension is indeed achievable for $\mathcal{O}(1)$-local Hamiltonians. Although the Sinha–Tong~\cite{sinha2025improved} bounds are inherently restricted to the short-time region determined by the operator norm, we derive new estimates that remain valid up to times of order \(1/\varepsilon\). These improved Bell sampling bounds establish that constant-locality Hamiltonians admit optimal $\Theta(1/\varepsilon)$-time intolerant certification.

\subsection{Open problems}
This work raises several natural questions about the fundamental limits of Hamiltonian certification and the scope of our techniques. We summarize some of these directions below.

\paragraph{Optimal certification without locality assumptions.}
Our main results exploit the locality of the underlying Hamiltonians, and in the $\mathcal{O}(1)$-local regime we show that optimal $\Theta(1/\varepsilon)$ total evolution time is achievable even without inverse or controlled evolution.  
A fundamental open problem is whether such optimal-time certification remains possible \emph{without} any locality assumptions on $H$ or~$H_0$.  
Previous optimal certification results for general (nonlocal) Hamiltonians, such as those of~\cite{gao2025quantum}, rely on access to inverse or controlled evolution.  
It is currently unknown whether optimal intolerant certification for arbitrary Hamiltonians can be achieved using only forward real-time dynamics.  
Resolving this question would clarify whether locality is merely a technical convenience in our analysis or whether it plays an essential role in enabling evolution-time--optimal protocols.

\paragraph{Optimal tolerant certification.}
In this paper we focus on the intolerant setting, where the algorithm must distinguish $H = H_0$ from the case $\|H - H_0\|_F \ge \varepsilon$, without any promise of an intermediate regime.  
A natural next step is to extend our approach to the tolerant formulation, in which one must distinguish $\|H - H_0\|_F \le \varepsilon_1$ from $\|H - H_0\|_F \ge \varepsilon_2$ for $0 \le \varepsilon_1 < \varepsilon_2$.  
For constant-locality Hamiltonians, it is plausible that the optimal total evolution time should scale as $\Theta((\varepsilon_2 - \varepsilon_1)^{-1})$, but our current analysis is tailored to the one-sided (intolerant) case and does not directly yield such a guarantee.  
Developing tolerant versions of our spectral conditions, and understanding the precise tradeoff between $\varepsilon_1$ and $\varepsilon_2$ in the local setting, remain interesting open problems.

\paragraph{Optimal dependence on the locality parameter.}
Our main result shows that for $k$-local Hamiltonians there exists a certification algorithm with total evolution time $\mathcal{O}(c^k/\varepsilon)$ and that, for $\mathcal{O}(1)$-local Hamiltonians, the $1/\varepsilon$ dependence is optimal.  
A natural question is whether one can obtain matching lower bounds that also capture the dependence on $k$, or alternatively whether the $c^k$ factor in our upper bound can be reduced.  
This requires new lower bound techniques that are sensitive to locality, going beyond simple one-qubit examples and exploiting many-body structure.  
Ideally, one would like a tight characterization of the optimal total evolution time for intolerant certification as a function of both $k$ and~$\varepsilon$, and to understand whether there exist regimes of $k$ where certification is strictly easier or harder than suggested by our current analysis.

\paragraph{Certification from local probes.}
Recent advances in Hamiltonian learning suggest that even severely restricted access, such as single-site or constant-size probes of a system undergoing time evolution, can still yield meaningful reconstruction guarantees~\cite{chen2025quantum}.  
This motivates a natural open problem for certification: whether one can efficiently certify a global $k$-local Hamiltonian using only the dynamics observed through a strictly local probe.  
A central question is whether there exists a quantitative tradeoff between probe size and certification complexity.  
For example, a constant-size probe may or may not suffice to achieve optimal $\Theta(1/\varepsilon)$ evolution time, and a larger probe may be fundamentally required to detect the spectral features that distinguish $H$ from $H_0$.  
Understanding this relationship remains open and would help bridge the gap between theoretically optimal certification protocols and experimentally realistic settings with limited local access.

\paragraph{Optimal sparsity testing.}
Our algorithm relies on random local basis changes, diagonalization in effective bases, and spectral gap arguments for the resulting diagonal Hamiltonians.  
These ingredients resemble those used in recent work on learning and testing $s$-sparse Hamiltonians.  
It is therefore natural to ask whether similar techniques can lead to optimal algorithms for Hamiltonian sparsity testing, for example with total evolution time $\mathcal{O}(s/\varepsilon)$ in intolerant settings (and $\mathcal{O}(s/(\varepsilon_2 - \varepsilon_1))$ in tolerant ones).  
Establishing such results would require adapting our eigenvalue-gap framework to properties defined in the Pauli coefficient space (such as sparsity) rather than purely spectral properties, and may shed further light on the precise relationship between learning and testing in the Hamiltonian setting.

\paragraph{Beyond Hamiltonians.}
Our analysis focuses on closed-system dynamics generated by time-independent Hamiltonians.
Many realistic devices, however, are better modeled by quantum channels or Lindbladian evolutions, such as Pauli channels induced by noise processes.
A natural direction is to investigate whether the techniques developed here, including randomized basis selection, effective diagonalization, and Bell-sampling based spectral tests, can be adapted to certify dynamical maps such as Pauli channels.
This would require extending our spectral conditions to appropriate distances on channels (for example, the diamond norm or time-constrained process distances) and identifying which aspects of the Hamiltonian setting are essential and which arise only from the specific model used here.

\section{Preliminaries}\label{sec:prelim}
\paragraph{Notation.}\label{para:notation}
We represent any traceless $n$-qubit Hamiltonian $H$ using its Pauli decomposition
\begin{equation}
    H \;=\; \sum_x \mu_x P_x,
\end{equation}
where $\mu_x\in\mathds{R}$ and $P_x$ ranges over the non-identity Pauli operators.

\medskip
We use several norms to quantify the size of a Hamiltonian.
The \emph{(unnormalized) Frobenius norm} is $\|H\|_{\mathrm{Frob}}
= \sqrt{\operatorname{Tr}(H^\dagger H)}$. The \emph{normalized Frobenius norm} is defined as
\begin{equation}
    \|H\|_{F}
    = \sqrt{\frac{\operatorname{Tr}(H^\dagger H)}{2^n}}
    = 2^{-n/2}\,\|H\|_{\mathrm{Frob}}
    = \left( \sum_x |\mu_x|^2 \right)^{1/2},
\end{equation}
which coincides with the $\ell_2$-norm of its Pauli coefficients.
The operator norm is $\|H\|_{\mathrm{op}}=\sup_{\|v\|=1}\|Hv\|$; for Hermitian $H$ this equals the spectral norm
$\max_i |\lambda_i|$. Since all Hamiltonians in this work are Hermitian, we use $\|\cdot\|_{\mathrm{op}}$ throughout.

\medskip
We write $[n]=\{1,2,\dots,n\}$. We also write $\mathds{F}_2$ for the finite field of size $2$, and $\mathds{F}_2^n$ for the $n$-dimensional vector space over $\mathds{F}_2$. For functions $f:\{-1,1\}^n\to\mathds{R}$ and $1\le p<\infty$, we define $\|f\|_p = \left(\mathds{E}[|f(x)|^p]\right)^{1/p}$, and $\|f\|_\infty = \max_{x\in\{-1,1\}^n} |f(x)|$.

\subsection{Bell sampling}\label{sec:bell-sample}
We recall the Bell basis and the Bell-sampling primitive used throughout this work. For two qubits, the Bell states are
\begin{align*}
    \ket{\sigma_{00}}=\frac{1}{\sqrt{2}}(\ket{00}+\ket{11}),\qquad
    \ket{\sigma_{01}}=\frac{1}{\sqrt{2}}(\ket{00}-\ket{11}), \\
    \ket{\sigma_{10}}=\frac{1}{\sqrt{2}}(\ket{01}+\ket{10}),\qquad
    \ket{\sigma_{11}}=\frac{1}{\sqrt{2}}(\ket{01}-\ket{10}).    
\end{align*}
The $2^n$-qubit Bell basis is given by all tensor products $\ket{\sigma_{\mathbf{s}}}=\ket{\sigma_{s_1}}\otimes\cdots\otimes\ket{\sigma_{s_n}}$, where $s_i\in\{00,01,10,11\}$. Equivalently, we may index outcomes by $\mathbf{s}\in\{0,1\}^{2n}$ via the standard identification
$\{00,01,10,11\}\cong\{0,1\}^2$ per qubit-pair.

\medskip
Given an $n$-qubit state $\rho$, \emph{Bell sampling}~\cite{montanaro2017learning,hangleiter2024bell}
means measuring $\rho^{\otimes 2}$ in the $n$-fold Bell basis
$\{\ket{\sigma_{\mathbf{s}}}\}_{\mathbf{s}\in\{0,1\}^{2n}}$.
The resulting outcome $\mathbf{s}$ is distributed according to the \emph{characteristic distribution}
$q_\rho(\mathbf{s})$.
For a pure state $\rho=\ket{\psi}\!\bra{\psi}$, this distribution admits the Pauli-moment form
\begin{equation}\label{eq:characteristic-distribution}
q_\rho(\mathbf{s})
=2^{-n}\,\bigl\langle \psi \big| P_{\mathbf{s}} \big| \psi \bigr\rangle^{2},
\end{equation}
where $P_{\mathbf{s}}$ denotes the $n$-qubit Pauli operator associated with~$\mathbf{s}$.
Thus, Bell sampling provides direct access to squared Pauli expectation values of the underlying state.

\medskip
A useful consequence of~\cref{eq:characteristic-distribution} is the following. If a unitary admits a Pauli expansion
\begin{equation}
    U=\sum_{x} U_x\,\sigma_x,    
\end{equation}
then Bell sampling applied to the state $U\ket{0^n}$ returns $\sigma_x$ with probability $|U_x|^2$. Sinha and Tong~\cite[Theorem~1]{sinha2025improved} observed that, when $\rho$ arises from Hamiltonian
evolution, the probability of obtaining the \emph{identity} outcome admits a spectral representation
in terms of eigenvalue gaps.
Let $H$ have eigenvalues $\lambda_1,\ldots,\lambda_N$ (with $N=2^n$), and define
\begin{align}\label{eq:bell-spectral}
I(t) \coloneqq \Pr[\,\text{Bell outcome } I\,]
= \frac{1}{4^n}\sum_{j,k}\cos((\lambda_j-\lambda_k)t).
\end{align}
We will use this expression to relate Bell-sampling statistics to the spectrum of~$H$.

\subsection{Analysis of Boolean functions}

We will need a few standard tools from the analysis of Boolean functions. For a comprehensive reference, see~\cite{o2014analysis}. We endow $\{-1,1\}^n$ with the uniform measure. All expectations and probabilities below are with respect to this measure unless stated otherwise.

\begin{definition}[(Fourier weight at degree $k$)]\label{def:fourier-weight}
Let $f:\{-1,1\}^n\to\mathds{R}$ and $0\le k\le n$. The (Fourier) weight of $f$ at degree $k$ is
\begin{equation}
    W^k[f] = \sum_{\substack{S\subseteq[n], \, |S|=k}} \widehat f(S)^2.
\end{equation}
If $f:\{-1,1\}^n\to\{-1,1\}$ is Boolean-valued, define the spectral distribution $\mathcal{S}_f$ on subsets of $[n]$ by $\Pr_{S\sim\mathcal{S}_f}[S=T] = \widehat f(T)^2$ for all $T\subseteq[n]$. Then $W^k[f] = \Pr_{S\sim\mathcal{S}_f}[|S|=k]$.
\end{definition}

\begin{theorem}[(Noise operator and hypercontractivity (Bonami--Beckner))]\label{thm:hypercontractivity}
Let $f:\{-1,1\}^n\to\mathds{R}$ and let $1\le p\le q\le\infty$.
For $\rho\in[-1,1]$, define the noise operator $T_\rho$ by
\begin{equation}
    (T_\rho f)(x) = \mathds{E}_{y\sim N_\rho(x)}[f(y)],    
\end{equation}
where $N_\rho(x)$ is the distribution on $\{-1,1\}^n$ obtained by independently setting, for each $i\in[n]$,
\[
y_i =
\begin{cases}
x_i, & \text{with probability } \frac{1+\rho}{2},\\
-x_i, & \text{with probability } \frac{1-\rho}{2}.
\end{cases}
\]
If $0\le \rho \le \sqrt{\frac{p-1}{q-1}}$, then $\|T_\rho f\|_q \le \|f\|_p$.
\end{theorem}

As a useful consequence of the $(2,q)$- and $(p,2)$-hypercontractivity inequalities, we will use the following moment bound for low-degree functions.

\begin{theorem}[(Generalization of the Bonami lemma)]\label{thm:low-degree-Lq}
Let $f:\{-1,1\}^n\to\mathds{R}$ have Fourier degree at most $k$.
Then for any $q\ge 2$, $\|f\|_q \le (q-1)^{k/2} \|f\|_2$.
\end{theorem}

\subsection{Useful lemmas}

We collect several probabilistic and linear-algebraic inequalities that will be used throughout the analysis.

\begin{lemma}[(Markov's inequality)]\label{lem:markov}
Let $Z\ge 0$ be a random variable and let $t>0$. Then
\begin{equation}
    \Pr[Z \ge t] \le \frac{\mathds{E}[Z]}{t}.
\end{equation}
\end{lemma}

\begin{lemma}[(Paley--Zygmund inequality)]\label{lem:p-z}
Let $Z\ge 0$ be a random variable with $\mathds{E}[Z^2]<\infty$ and let $0\le \theta\le 1$. Then
\begin{equation}
    \Pr\left[ Z > \theta\,\mathds{E}[Z] \right]
    \ge
    (1-\theta)^2 \cdot \frac{\mathds{E}[Z]^2}{\mathds{E}[Z^2]}.
\end{equation}
\end{lemma}

We use~\cref{lem:p-z} to lower-bound the probability that a suitable nonnegative random variable is not too small in terms of its first two moments. In particular, we apply it to squared eigenvalue differences, taking $Z$ to be a quadratic function of the spectrum and choosing $\theta$ as a fixed constant.

\begin{lemma}[(Hoffman--Wielandt inequality~\cite{hoffman1953the})]\label{lem:hoffman-wielandt}
Let $A,B\in\mathds{C}^{d\times d}$ be normal matrices, and let
$\{\lambda_i(A)\}_{i=1}^d$ and $\{\lambda_i(B)\}_{i=1}^d$ denote their eigenvalues.
Then there exists a permutation $\pi$ of $[d]$ such that
\begin{equation}
    \frac{1}{d}\sum_{i=1}^d \big|\lambda_i(A)-\lambda_{\pi(i)}(B)\big|^2
    \le \|A-B\|_{F}^2,
\end{equation}
If $A$ and $B$ are Hermitian, the same bound holds.
\end{lemma}

\section{Sufficient conditions for Hamiltonian certification}\label{sec:sufficient-conditions}

In this section, we analyze the Bell-sampling identity-outcome probability associated with the
evolution operator $e^{-iHt}$ and identify structural conditions under which optimal
$\Theta(1/\varepsilon)$-time Hamiltonian certification can be achieved. Our goal is to clarify how the
spectral structure of a Hamiltonian governs the detectability of deviations from a reference
Hamiltonian.

\medskip
Let $I(t)$ denote the probability of obtaining the identity outcome in Bell sampling applied to the
evolved state (as defined via the spectral representation in~\cref{sec:bell-sample}).
By the spectral representation in~\cref{sec:bell-sample}, this identity-outcome probability depends
only on eigenvalue \emph{differences} (see~\cref{eq:bell-spectral}). This motivates a quantitative
measure of how much of the spectrum contributes nontrivially to the Bell-sampling signal.

\begin{definition}[(Proportion of $\varepsilon$-separated eigenvalue pairs)]
Let $X\in\mathds{C}^{N\times N}$ be Hermitian with eigenvalues $\lambda_1,\ldots,\lambda_N$.
For $\varepsilon>0$, define
\begin{equation}
    \Lambda(X,\varepsilon) \coloneqq \frac{1}{N^2}\sum_{j,k:\,|\lambda_j-\lambda_k|\ge \varepsilon} 1.
\end{equation}
\end{definition}

A large value of $\Lambda(X,\varepsilon)$ indicates that many eigenvalue gaps contribute oscillatory
terms to the probability of obtaining the identity outcome in Bell sampling. In particular, if
$\Lambda(X,\varepsilon)$ is bounded below by a positive constant, then this identity-outcome
probability cannot remain close to $1$ throughout $t\in[0,2/\varepsilon]$, since destructive
interference among the cosine terms forces a noticeable decrease.

\medskip
To obtain an optimal-time certification guarantee for general Hamiltonians, it is therefore
sufficient to show that whenever $\|H\|_F$ is at least $\varepsilon$, the proportion of eigenvalue pairs
separated by at least $\varepsilon$ is bounded below by a positive constant. Equivalently, it suffices
to establish the existence of a constant $c>0$ such that $\Lambda(H,\|H\|_F)\ge c$ for all
Hamiltonians under consideration.

\subsection{A sufficient spectral condition}
We now formalize how a lower bound on $\Lambda(H,\varepsilon)$ yields an optimal-time certification
procedure.

\begin{lemma}[(Sufficient spectral condition)]\label{lem:spectral-condition}
Let $H$ be an $n$-qubit Hamiltonian and fix $\varepsilon>0$.
Suppose there exists a constant $d>0$ such that $\Lambda(H,\varepsilon) \ge d$.
Then there is a randomized procedure that, using total evolution time
$\mathcal{O}(\log(1/\delta)/\varepsilon)$ and $\mathcal{O}(\log(1/\delta))$ samples of $I(t)$ for
uniformly random $t \in [0,2/\varepsilon]$, finds a time $t$ satisfying
\begin{equation}
    I(t) \le 1 - \frac{d}{4}
\end{equation}
with probability at least $1-\delta$.
In particular, suppose that for the case $\|H\|_F \ge \varepsilon$, $\Lambda(H,\varepsilon)$ is always bounded below by a constant. Then, for constant $\delta$, one can distinguish between $\|H\|_F = 0$ and $\|H\|_F \ge \varepsilon$ using total evolution time
$\mathcal{O}(1/\varepsilon)$.
\end{lemma}

\begin{proof}
Let $t$ be uniformly random in $[0,2/\varepsilon]$, and write $\Delta_{jk} := \lambda_j - \lambda_k$.  
From the spectral representation of Bell sampling,
\begin{align}
    I(t)
    & = \frac{1}{4^n} \sum_{j,k} \cos(\Delta_{jk} t) \nonumber \\
    & = \mathds{E}_{\Delta_{jk}}\left[ \cos(\Delta_{jk} t) \right].
\end{align}
Define $I_{\ge \varepsilon}(t)
    := \mathds{E}_{|\Delta_{jk}|\ge \varepsilon}\left[ \cos(\Delta_{jk} t) \right]$, and let $\Lambda := \Lambda(H,\varepsilon)$ denote the proportion of eigenvalue pairs with $|\Delta_{jk}| \ge \varepsilon$.  
Then
\begin{equation}
    I(t) \le (1-\Lambda) + \Lambda\, I_{\ge \varepsilon}(t).
\end{equation}

\medskip
For any pair $(j,k)$ with $|\Delta_{jk}| \ge \varepsilon$, the average of $\cos(\Delta_{jk} t)$ over $t \sim \mathrm{Unif}[0,2/\varepsilon]$ is
\begin{equation}
    \mathds{E}_t\left[\cos(\Delta_{jk} t)\right]
        = \frac{\varepsilon}{2}\cdot\frac{\sin(2\Delta_{jk}/\varepsilon)}{\Delta_{jk}},
\end{equation}
which satisfies
\begin{equation}
    \left|\mathds{E}_t[\cos(\Delta_{jk} t)]\right|
        \le \frac{\varepsilon}{2|\Delta_{jk}|}
        \le \frac{1}{2},
\end{equation}
since $|\Delta_{jk}| \ge \varepsilon$.  
For pairs with $|\Delta_{jk}| < \varepsilon$, we simply use the trivial bound $\cos(\Delta_{jk} t) \le 1$.  
Averaging termwise yields $\mathds{E}_t[I_{\ge \varepsilon}(t)] \le 1/2$.

\medskip
Now consider the random variable $X := I_{\ge \varepsilon}(t)$ with $t \sim \mathrm{Unif}[0,2/\varepsilon]$.  
We have $0 \le X \le 1$ for all $t$, and $\mathds{E}_t[I_{\ge \varepsilon}(t)] \le 1/2$ implies $\mathds{E}[X] \le 1/2$.  
Let $A := \{t \in [0,2/\varepsilon] : X(t) > 3/4\}$, $B := [0,2/\varepsilon] \setminus A$, and let $\mu$ be the uniform measure on the interval. Write $\alpha := \mu(A)$ and $\beta := \mu(B) = 1-\alpha$.

\medskip
Since $X(t) \ge 3/4$ for $t \in A$ and $X(t) \ge 0$ everywhere, we have
\begin{align}
    \mathds{E}[X]
    & = \int_{0}^{2/\varepsilon} X(t)\, d\mu(t) \nonumber \\
    & \ge \int_{A} X(t)\, d\mu(t)
    \ge \frac{3}{4}\alpha.
\end{align}
Combining this with $\mathds{E}[X] \le 1/2$ gives $\alpha \le {2}/{3}$, and hence $\beta = 1-\alpha \ge {1}/{3}$.

\medskip
Therefore, with probability at least $1/3$ over a uniformly random $t \in [0,2/\varepsilon]$,
\begin{align}
    I(t)
    &\le (1-\Lambda) + \Lambda\, I_{\ge \varepsilon}(t) \nonumber \\
    & \le (1-\Lambda) + \Lambda \cdot \frac{3}{4} \nonumber \\
    & = 1 - \frac{\Lambda}{4}.
\end{align}
Using the assumption $\Lambda \ge d$, we conclude that
\begin{equation}
    I(t) \le 1 - \frac{d}{4}
\end{equation}
with probability at least $1/3$.

\medskip
To find such a time $t$ with failure probability at most $\delta$, it suffices to draw  
$m = \Theta(\log(1/\delta))$ independent samples $t_1,\dots,t_m$ uniformly from $[0,2/\varepsilon]$.  
Since $\beta \ge 1/3$, the probability that none of these samples lies in $B$ is at most  
$(1 - 1/3)^m \le \delta$ for an appropriate constant.  
Each sample uses evolution time at most $2/\varepsilon$, so the total evolution time is  
$\mathcal{O}(\log(1/\delta)/\varepsilon)$, which becomes $\mathcal{O}(1/\varepsilon)$ for constant~$\delta$. Finding such $t$ allows us to discriminate the two cases $\|H\|_F=0$ and $\|H\|_F\ge\varepsilon$, because $I(t)=1$ always holds for $\|H\|_F=0$.
\end{proof}

This spectral condition holds for certain structured families of Hamiltonians, such as those
diagonal in the Pauli-$Z$ basis, whose spectra are sufficiently spread out to contain many
well-separated eigenvalue gaps. Establishing an analogous property for general constant-local
Hamiltonians is the key technical step required for optimal certification guarantees. In the
sections that follow, we verify this condition for Pauli-$Z$-diagonal Hamiltonians.

\subsection{Structure of $Z$-diagonal Hamiltonians}
We begin by showing that $Z$-diagonal $k$-local Hamiltonians satisfy the eigenvalue-gap condition
in~\cref{lem:spectral-condition}.

\begin{proposition}[(Eigenvalue-gap bound for $Z$-diagonal $k$-local Hamiltonians)]\label{pro:Z-basis}
Let $H$ be an $n$-qubit, $k$-local, traceless Hamiltonian of the form
\begin{equation}
    H = \sum_{P\in\{I,Z\}^{\otimes n},\, |P|\le k} \alpha_P P.
\end{equation}
Then the proportion of eigenvalue pairs seprated by at least $\|H\|_F$ satisfies
\begin{equation}
    \Lambda(H,\|H\|_F) \ge \frac14 \cdot 9^{-k}.
\end{equation}
\end{proposition}

\begin{proof}
Define a map $g:\{I,Z\}\to\{0,1\}$ by $g(I)=0$ and $g(Z)=1$.  
Each Pauli string $P\in\{I,Z\}^{\otimes n}$ corresponds to a bit string $p=g(P)\in\{0,1\}^n$, and one checks that
\begin{equation}
    P = \sum_{s\in\{0,1\}^n} (-1)^{s\cdot p}\,\ket{s}\bra{s},
\end{equation}
where $s\cdot p$ denotes the standard inner product over $\mathds{F}_2^n$.  
Writing $H$ in this basis, define
\begin{equation}
    f(s) = \sum_{\substack{p\in\{0,1\}^n,\, |p|\le k}} (-1)^{s\cdot p} \, \alpha_p,
\end{equation}
so that
\begin{equation}
    H = \sum_{s\in\{0,1\}^n} f(s)\ket{s}\bra{s}.
\end{equation}
Thus the values $f(s)$, for $s\in\{0,1\}^n$, are precisely the eigenvalues of $H$. Since $H$ is traceless, we have
\begin{equation}
    \mathds{E}_s[f(s)] = \frac{1}{2^n}\sum_{s} f(s) = 0.
\end{equation}
Consider independent, uniformly random $s,t \in \{0,1\}^n$, and define $F(s,t) = f(s) - f(t)$. We regard $F$ as a random variable over $(s,t)$ and analyze its second and fourth moments.

\medskip
\noindent\textit{Second moment.}
We first compute $\mathds{E}[f(s)^2]$.  Using orthogonality of characters on the Boolean cube,
\begin{align}
    \mathds{E}_s[f(s)^2]
    &= \frac{1}{2^n}\sum_{s}\left(\sum_{p_1} (-1)^{s\cdot p_1}\alpha_{p_1}\right)
                         \left(\sum_{p_2} (-1)^{s\cdot p_2}\alpha_{p_2}\right) \nonumber \\
    &= \frac{1}{2^n}\sum_{p_1,p_2}\alpha_{p_1}\alpha_{p_2}
       \sum_{s} (-1)^{s\cdot(p_1+p_2)}.
\end{align}
For $p_1\neq p_2$, exactly half of the strings $s$ satisfy $s\cdot(p_1+p_2)\equiv 0 \pmod{2}$ and half satisfy $s\cdot(p_1+p_2)\equiv 1 \pmod{2}$, so
\begin{equation}
    \sum_{s} (-1)^{s\cdot(p_1+p_2)} = 0.
\end{equation}

\medskip
For $p_1 = p_2 = p$, we have $\sum_{s} (-1)^{s\cdot 0} = 2^n$.  
Hence
\begin{equation}
    \mathds{E}_s[f(s)^2] = \sum_{p} \alpha_p^2.
\end{equation}
On the other hand, by the normalization of the Frobenius norm in the Pauli basis,
\begin{align}
    \|H\|_F^2 & = 2^{-n}\operatorname{Tr}(H^2) \nonumber\\
    & = \sum_{P} \alpha_P^2 = \sum_{p} \alpha_p^2,
\end{align}
so $\mathds{E}_s[f(s)^2] = \|H\|_F^2$.

\medskip
Using independence of $s$ and $t$ and the fact that $\mathds{E}[f]=0$, we obtain
\begin{align}
    \mathds{E}_{s,t}[F(s,t)^2]
    &= \mathds{E}_{s,t}[(f(s)-f(t))^2] \nonumber \\
    &= \mathds{E}_s[f(s)^2] + \mathds{E}_t[f(t)^2] - 2\,\mathds{E}_s[f(s)]\,\mathds{E}_t[f(t)] \nonumber \\
    &= 2\,\mathds{E}_s[f(s)^2] \nonumber \\
    &= 2\|H\|_F^2.
\end{align}
Thus $\mathds{E}[F^2] = 2\|H\|_F^2$.

\medskip
\noindent\textit{Fourth moment.}
We next bound $\mathds{E}[F^4]$ from above.  
The function $f(s)$ is a multilinear polynomial of degree at most $k$ in the $n$ Boolean variables $s\in\{0,1\}^n$ (or equivalently $\{\pm1\}^n$ after recoding).  
Consequently, $F(s,t) = f(s) - f(t)$ is a degree-$k$ polynomial in the $2n$ Boolean variables $(s,t)$.  
By hypercontractivity for Boolean functions (\cref{thm:hypercontractivity,thm:low-degree-Lq}), for any $q\ge 2$,
\begin{equation}
    \|F\|_q \le (q-1)^{k/2} \|F\|_2.
\end{equation}
Specializing to $q=4$ gives
\begin{align}
    \mathds{E}[F^4] &= \|F\|_4^4 \nonumber\\
    & \le 9^k \|F\|_2^4 \nonumber\\
    & = 9^k (\mathds{E}[F^2])^2.
\end{align}

\medskip
\noindent\textit{Applying Paley--Zygmund.}
Let $Y := F^2$, which is a nonnegative random variable.  
Applying Paley--Zygmund inequality (\cref{lem:p-z}) to $Y$ with parameter $\theta = 1/2$, we obtain
\begin{align}
    \Pr\left[ Y \ge \frac12\,\mathds{E}[Y] \right]
    & \ge
    \left(1 - \frac12\right)^2 \frac{\mathds{E}[Y]^2}{\mathds{E}[Y^2]} \nonumber \\
    & =
    \frac14 \frac{(\mathds{E}[F^2])^2}{\mathds{E}[F^4]}.
\end{align}
Using the moment bounds derived above,
\begin{align}
    \Pr\left[ Y \ge \frac12\, \mathds{E}[Y] \right]
    &\ge \frac14 \cdot \frac{(\mathds{E}[F^2])^2}{9^k (\mathds{E}[F^2])^2} \nonumber \\
    & = \frac14 \cdot 9^{-k}.
\end{align}
Since $\mathds{E}[Y] = \mathds{E}[F^2] = 2\|H\|_F^2$, the event $Y \ge \mathds{E}[Y]/2$ is exactly the event $F^2 \ge \|H\|_F^2$. Hence
\begin{equation}
    \Pr\left[ F^2 \ge \|H\|_F^2 \right] \ge \frac14 \cdot 9^{-k},
\end{equation}
or equivalently
\begin{equation}
    \Pr\left[|F| \ge \|H\|_F \right] \ge \frac14 \cdot 9^{-k}.
\end{equation}

\medskip
By construction, sampling $(s,t)$ uniformly from $\{0,1\}^n \times \{0,1\}^n$ is equivalent to choosing an ordered pair of eigenvalues $(\lambda_j,\lambda_k)$ of $H$ uniformly at random (with replacement).  
Thus the probability on the left-hand side is exactly the fraction of eigenvalue pairs $(\lambda_j,\lambda_k)$ such that $|\lambda_j - \lambda_k| \ge \|H\|_F$.  
This fraction is precisely $\Lambda(H,\|H\|_F)$, so we obtain
\begin{equation}
    \Lambda(H,\|H\|_F) \ge \frac14 \cdot 9^{-k}
\end{equation}
as claimed.
\end{proof}

\subsection{Diagonal Hamiltonians}

The preceding argument applies verbatim to any Hamiltonian that is diagonal in the computational basis.  
Such a Hamiltonian can be written as
\begin{align}
    H = \sum_{P \in S} \alpha_P P, \quad
    S = \bigotimes_{i=1}^n \{ I,\, Q^{(i)} \},
\end{align}
where $Q^{(i)}$ is an arbitrary single-qubit Pauli operator acting on the $i$-th qubit.

\medskip
For every diagonal operator set $S$, there exists a unitary $U_S$ that maps the standard $Z$-diagonal Pauli operators $\{I,Z\}^{\otimes n}$ onto $S$.  
Therefore, any diagonal $k$-local Hamiltonian can be expressed as
\begin{equation}
    H = U_S \left( \sum_{P \in \{I,Z\}^{\otimes n},\, |P|\le k} \alpha_P P \right) U_S^\dagger
\end{equation}
for appropriate coefficients $\alpha_P$.  
Since the eigenvalues of a Hamiltonian are invariant under unitary conjugation, eigenvalue-gap bound for $Z$-diagonal $k$-local Hamiltonians (\cref{pro:Z-basis}) applies directly, implying that every diagonal $k$-local Hamiltonian satisfies
\begin{equation}
    \Lambda(H,\|H\|_F) \ge \frac{1}{4} \cdot 9^{-k}.
\end{equation}
Combining this with sufficient spectral condition (\cref{lem:spectral-condition}), we conclude that any diagonal $k$-local Hamiltonian admits an intolerant certification procedure using total evolution time of order $\mathcal{O}(9^k/\varepsilon)$ with constant success probability.

\section{Certification algorithm for general local Hamiltonians}

In this section we describe and analyze our main certification algorithm for general $k$-local Hamiltonians.
The algorithm has three components:
\begin{enumerate}
    \item a randomized choice of a diagonal Pauli basis that preserves a non-negligible fraction of the (normalized) Frobenius norm,
    \item a random twirling procedure that approximately projects onto this diagonal subspace while suppressing the remaining terms, and
    \item the Bell-sampling spectral condition from~\cref{lem:spectral-condition}, together with Trotterization to implement the required evolutions.
\end{enumerate}
Throughout, we assume oracle access to the forward-time evolutions $e^{-itH}$ for the unknown Hamiltonian $H$ and $e^{-itH_0}$ for the known reference Hamiltonian $H_0$.

\subsection{Randomized diagonal basis selection}\label{sec:rand-diag-sel}

We first show that a random choice of local Pauli bases picks out a diagonal sub-Hamiltonian that retains a nontrivial fraction of the normalized Frobenius norm.

\begin{proposition}[(Randomized diagonal basis selection)]\label{pro:random-basis}
Let $H$ be an $n$-qubit, $k$-local, traceless Hamiltonian with Pauli expansion
\begin{equation}
    H = \sum_{P\in\{I,X,Y,Z\}^{\otimes n},\, |P|\le k} \alpha_P P.
\end{equation}
For independent random choices $Q^{(1)},\dots,Q^{(n)} \in \{X,Y,Z\}$, define
\begin{equation}
    S := \bigotimes_{i=1}^n \{I,Q^{(i)}\},\qquad
    H_{\mathrm{eff}} := \sum_{P\in S,\, |P|\le k} \alpha_P P.
\end{equation}
Then
\begin{equation}
    \Pr\left[\|H_{\mathrm{eff}}\|_F \ge \frac{\|H\|_F}{\sqrt{2} \cdot 3^{k/2}}\right] \ge \frac{1}{4\cdot 3^k}.
\end{equation}
\end{proposition}

\begin{proof}[Proof sketch]
We provide a brief sketch; the full proof appears in~\cref{app:proof-random-basis}.

\medskip
\noindent For each Pauli string $P$ with weight $|P|\le k$, the probability that $P$ lies in the random subspace $S$
equals $3^{-|P|}$, since each non-identity tensor factor of $P$ must match the corresponding $Q^{(i)}$.
Consequently,
\begin{align}
    \mathds{E}[\|H_{\mathrm{eff}}\|_F^2]
    & = \sum_{P} \alpha_P^2\, \Pr[P\in S]\nonumber \\
    & \ge 3^{-k} \sum_{P} \alpha_P^2 \nonumber\\
    & = 3^{-k}\|H\|_F^2.
\end{align}
A second-moment bound yields $\mathds{E}[\|H_{\mathrm{eff}}\|_F^4]\le 3^k\,
(\mathds{E}[\|H_{\mathrm{eff}}\|_F^2])^2$.
Applying the Paley--Zygmund inequality to the nonnegative random variable $\|H_{\mathrm{eff}}\|_F^2$ (with $\theta=1/2$) gives
\begin{equation}
    \Pr\left[\|H_{\mathrm{eff}}\|_F^2 \ge \frac12\,\mathds{E}[\|H_{\mathrm{eff}}\|_F^2]\right] \ge \frac{1}{4\cdot 3^k}.    
\end{equation}
Combining this with the lower bound on $\mathds{E}[\|H_{\mathrm{eff}}\|_F^2]$ yields the claim.
\end{proof}

\subsection{Extracting the effective diagonal Hamiltonian}\label{sec:ext-eff-diag}

We next show how to suppress the components of $H$ that lie outside the chosen diagonal subspace $S$, while preserving $H_{\mathrm{eff}}$.

\begin{proposition}[(Random twirling toward the diagonal subspace)]\label{pro:z-twirl}
Fix a choice of $Q^{(1)},\dots,Q^{(n)}\in\{X,Y,Z\}$ and let $S = \bigotimes_{j=1}^n\{I,Q^{(j)}\}$.  Decompose
\begin{equation}
    H_1 := H - H_0 = H_{\mathrm{eff}} + H_1',
\end{equation}
where $H_{\mathrm{eff}} = \sum_{P\in S} \alpha_P P$ and $H_1' = \sum_{P\notin S} \alpha_P P$.  
Define recursively
\begin{align}
    H_i = \frac{1}{2}(H_{i-1}+P_i H_{i-1} P_i), \quad
    H_i' = \frac{1}{2}(H_{i-1}'+P_i H_{i-1}' P_i),
\end{align}
for $i\ge 2$, where each $P_i$ is chosen independently and uniformly at random from $S$.  
Then $H_i = H_{\mathrm{eff}} + H_i'$ for all $i$, and for every integer $T\ge 1$,
\begin{equation}\label{eq:prop-z}
    \Pr\left[\|H_T'\|_F \le \frac{2}{2^{T/2}}\|H\|_F\right] \ge \frac{3}{4}.
\end{equation}
\end{proposition}

\begin{proof}[Proof sketch]
We provide a brief sketch; the full proof appears in~\cref{app:proof-z-twirl}.

\medskip
\noindent Since $S=\bigotimes_{j=1}^n\{I,Q^{(j)}\}$ is abelian, every $P_i\in S$ commutes with every $P\in S$, hence the averaging map
$X\mapsto \frac12(X+P_iXP_i)$ fixes $H_{\mathrm{eff}}$ and yields $H_i=H_{\mathrm{eff}}+H_i'$ by induction.

\medskip
For $P\notin S$, there is a qubit where $P$ anticommutes with $Q^{(j)}$, so a uniformly random $P_i\in S$ anticommutes with $P$ with probability $1/2$.
In that case $\frac12(P+P_iPP_i)=0$, so each such Pauli coefficient survives $T$ steps with probability $2^{-T}$.
Therefore $\mathds{E}[\|H_T'\|_F^2]=2^{-T}\|H_1'\|_F^2\le 2^{-T}\|H\|_F^2$, and Markov's inequality implies~\cref{eq:prop-z}
\end{proof}

Combining~\cref{pro:random-basis,pro:z-twirl}, and choosing $T=\Theta(k)$ so that $2^{T} \ge 2^{10}\cdot 3^{5k}$, we obtain with constant probability (over the choices of $Q^{(j)}$ and the twirling sequence) that
\begin{equation}
    \|H_{\mathrm{eff}}\|_F \ge \frac{\|H-H_0\|_F}{\sqrt{2}\cdot3^{k/2}},\qquad
    \|H_T'\|_F \le \frac{1}{8\sqrt{2}\cdot3^k}\,\|H_{\mathrm{eff}}\|_F.
\end{equation}
Hence the final Hamiltonian $H_T = H_{\mathrm{eff}} + H_T'$ is close to $H_{\mathrm{eff}}$ in Frobenius norm (see~\cref{subsec:stab,subsec:put})
and inherits its spectral-gap structure up to a controlled perturbation.

\subsection{Stability of eigenvalue gaps under perturbations}\label{subsec:stab}

We derive quantitative bounds on how eigenvalue-gap statistics behave under perturbations of small Frobenius norm.

\begin{proposition}[(Stability of eigenvalue-gap proportion)]\label{pro:eig-stability}
Let $A,B \in \mathds{C}^{N\times N}$ be Hermitian matrices.  
Suppose that for some $\varepsilon>0$ and $p\in[0,1]$, $\Lambda(A,\varepsilon) \ge p$,  $\|B\|_F \le q\varepsilon$ for some $q\ge 0$. Then
\begin{equation}
    \Lambda(A+B,\varepsilon/2) \ge \max\{0,\, p - 32q^2\}.
\end{equation}
\end{proposition}

\begin{proof}
Let $\{\lambda_i\}$ and $\{\mu_i\}$ denote the eigenvalues of $A$ and $A+B$, ordered so that the Hoffman--Wielandt inequality (\cref{lem:hoffman-wielandt}) applies:
\begin{equation}
    \frac{1}{N}\sum_{i=1}^N (\lambda_i - \mu_i)^2 
    \le \|B\|_F^2 
    \le q^2 \varepsilon^2.
\end{equation}
Define $T := \left\{i : |\lambda_i - \mu_i| > \varepsilon/4 \right\}$. Then
\begin{align}
    \frac{|T|}{N} \left(\frac{\varepsilon}{4}\right)^2
    &\le \frac{1}{N}\sum_{i\in T} (\lambda_i - \mu_i)^2 \nonumber\\
    &\le \frac{1}{N}\sum_{i=1}^N (\lambda_i - \mu_i)^2 \nonumber\\
    &\le q^2 \varepsilon^2,
\end{align}
which gives $|T|/N \le 16 q^2$.

\medskip
For indices $i,j \notin T$, we have
\begin{align}
    |\mu_i - \mu_j|
    &\ge |\lambda_i - \lambda_j|
        - |\lambda_i - \mu_i|
        - |\lambda_j - \mu_j| \\
    &\ge |\lambda_i - \lambda_j| - \frac{\varepsilon}{2}.
\end{align}
Hence any pair $(i,j)$ with $|\lambda_i - \lambda_j| \ge \varepsilon$ and $i,j \notin T$ satisfies $|\mu_i - \mu_j| \ge \varepsilon/2$.

\medskip
By assumption, at least a proportion $p$ of all pairs $(i,j)$ obey $|\lambda_i - \lambda_j| \ge \varepsilon$.  
Among all pairs, at most a fraction
\begin{equation}
    1 - (1 - 16q^2)^2 \le 32q^2
\end{equation}
involve an index in $T$.  
Therefore at least a proportion $p - 32q^2$ of all pairs have both indices outside $T$ and hence satisfy $|\mu_i - \mu_j| \ge \varepsilon/2$. By definition, this means $\Lambda(A+B,\varepsilon/2) \ge p - 32q^2$. Since $\Lambda(\cdot,\cdot) \ge 0$, this yields the stated bound.
\end{proof}

\subsection{Implementing the effective evolution by Trotterization}\label{sec:trott}

To apply the Bell-sampling analysis from~\cref{sec:sufficient-conditions} to $H_T$, we require an implementation of time evolution under $H_T = H_{\mathrm{eff}} + H_T'$ using only access to $e^{-itH}$ and $e^{-itH_0}$.  
This can be achieved using a symmetric Trotterization procedure.

\begin{theorem}[(Trotterization)]\label{thm:trotter}
Let $t>0$ and $\varepsilon_{\mathrm{Trott}}>0$, and let $H_T$ be the Hamiltonian defined in~\cref{pro:z-twirl}. Then one can implement $e^{-itH_T}$ up to diamond-norm error $\varepsilon_{\mathrm{Trott}}$ using  
\begin{equation}
    \ell = \mathcal{O}\left(2^T \sqrt{\frac{t^3}{\varepsilon_{\mathrm{Trott}}}}\right)
\end{equation}
queries to the time-evolution operators of $H$ and $H_0$. Moreover, the total evolution time under $H$ and $H_0$ required to implement $V$ is $\mathcal{O}(t)$.
\end{theorem}

\begin{proof}
This follows from the standard second-order Trotter--Suzuki product formula.  
The operator-norm error bound follows from the Baker--Campbell--Hausdorff expansion together with known analyses of higher-order Trotter formulas; a detailed argument is provided in~{\cite[Lemma~3.3]{abbas2025nearly}}. Since each Trotter step uses evolution time $t/\ell$, the total evolution time is $\ell \cdot (t/\ell) = t$ for each of $H$ and $H_0$, up to constant factors.
\end{proof}

We now analyze how Trotterization affects Bell sampling.  
Let $V$ denote the Trotterized approximation to $e^{-itH_T}$. Then
\begin{align}
    \bigl|\Pr[\,\text{Bell outcome } I \text{ from } e^{-itH_T}\,]
      - \Pr[\,\text{Bell outcome } I \text{ from } V\,]\bigr|
      &= \left|\frac{1}{2^n}\Tr(e^{-itH_T} - V)\right| \nonumber \\
    &= \left|\frac{1}{2^n}\sum_j \bra{j}(e^{-itH_T}-V)\ket{j}\right| \nonumber \\
    &\le \frac{1}{2^n}\sum_j \|e^{-itH_T} - V\|_{\diamond} \nonumber \\
    &\le \varepsilon_{\mathrm{Trott}}.
\end{align}
Thus, the impact of Trotterization on Bell-sampling statistics is bounded by $\varepsilon_{\mathrm{Trott}}$.

\medskip
Hence, the effective per-unit-time query rate is  $\ell/t=\mathcal{O}\left(2^T \sqrt{{t}/{\varepsilon_{\mathrm{Trott}}}}\right)$, and the total query complexity in our algorithm follows by multiplying this quantity by the total evolution time (specified in~\cref{subsec:put}).

\subsection{Putting everything together}\label{subsec:put}

We now describe the full certification procedure (\cref{alg:main-algorithm}) and state its performance guarantee.
Concrete values for the numerical constants appearing in the algorithm are provided in~\cref{lem:constants}.

\begin{algorithm}[h]
\DontPrintSemicolon
\caption{\textsc{IntolerantCert-LocalH}\,$(H_0, H, \varepsilon, \delta, k)$}
\label{alg:main-algorithm}
\KwIn{Description of $H_0$, oracle access to $e^{-itH}$, parameters $\varepsilon>0$, $\delta\in(0,1)$, locality parameter $k$ (with universal constants $C_1,C_2,C_3,C_4,c_0$).}
\KwOut{\textsf{ACCEPT} or \textsf{REJECT}.}

\smallskip
Set $R \gets C_1 \cdot 3^k \log(1/\delta)$\;
\For{$r = 1$ \KwTo $R$}{
    \tcp{1. Random diagonal basis selection (\cref{sec:rand-diag-sel})}
    Sample $Q^{(1)},\dots,Q^{(n)} \in \{X,Y,Z\}$ independently and uniformly\;
    Set $S \gets \bigotimes_{j=1}^n \{I,Q^{(j)}\}$\;

    \medskip
    \tcp{2. Random twirling to obtain $H_T$ (\cref{sec:ext-eff-diag})}
    Set $T \gets C_2 \cdot k$\;
    Conceptually set $H_1 \gets H - H_0$\;
    \For{$i = 2$ \KwTo $T$}{
        Sample $P_i \in S$ uniformly at random\;
        Update $H_i \gets \frac{1}{2}(H_{i-1} + P_i H_{i-1} P_i)$\;
    }

    \medskip
    \tcp{3. Implement $e^{-itH_T}$ by Trotterization (\cref{sec:trott})}
    Set $b \gets C_3 \cdot 3^{k/2}/\varepsilon$\;
    Sample $t$ uniformly from $[0,b]$\;
    Use~\cref{thm:trotter} to implement a unitary $V$ satisfying
    $\|V - e^{-itH_T}\|_{\diamond} \le \varepsilon_{\mathrm{Trott}}$\;

    \medskip
    \tcp{4. Bell sampling}
    Set $m \gets C_4\cdot 9^k$\;
    Let $\widehat I_r(t)$ denote the fraction of identity outcomes among $m$ Bell samples using $V$\;

    \medskip
    \If{$\widehat I_r(t) \le 1 - c_0/9^k$}{
    \smallskip
        \Return \textsf{REJECT}
    }
}

\medskip
\Return \textsf{ACCEPT}
\smallskip
\end{algorithm}

\begin{theorem}[(Main certification algorithm)]\label{thm:main-algorithm}
Let $H$ and $H_0$ be $n$-qubit, $k$-local, traceless Hamiltonians.
For any $\varepsilon>0$ and $\delta\in(0,1)$, \cref{alg:main-algorithm} solves \cref{problem:intolerant} using total evolution time
\begin{equation}
    \mathcal{O}\left(\frac{3^k\cdot 3^{k/2}\cdot 9^k \log(1/\delta)}{\varepsilon}\right)
    = \mathcal{O}\left(\frac{c^k \log(1/\delta)}{\varepsilon}\right)
\end{equation}
for some universal constant $c>1$. For constant locality $k=\mathcal{O}(1)$ and constant failure probability, the total evolution time is $\Theta(1/\varepsilon)$.
\end{theorem}

\begin{proof}
If $H = H_0$, then $H_1 = 0$ and hence $H_{\mathrm{eff}} = 0$ and $H_T' = 0$ in every round, so $e^{-itH_T}=I$.
Choosing $\varepsilon_{\mathrm{Trott}}$ sufficiently small and taking a union bound shows that the algorithm outputs \textsf{ACCEPT} with probability at least $1-\delta$.

\medskip
Now suppose that $\|H - H_0\|_F \ge \varepsilon$.
By randomized diagonal basis selection (\cref{pro:random-basis}), with probability at least $1/(4\cdot 3^k)$ over the random basis,
\begin{align}
    \|H_{\mathrm{eff}}\|_F & \ge \frac{\|H-H_0\|_F}{\sqrt{2}\cdot 3^{k/2}} \nonumber \\
    & \ge \frac{\varepsilon}{\sqrt{2}\cdot 3^{k/2}}.
\end{align}

\medskip
Since $H_{\mathrm{eff}}$ is diagonal in that basis, eigenvalue-gap bound for $Z$-diagonal $k$-local Hamiltonians (\cref{pro:Z-basis}) yields
\begin{equation}
    \Lambda(H_{\mathrm{eff}}, \|H_{\mathrm{eff}}\|_F) \ge \frac{1}{4}\cdot 9^{-k}.    
\end{equation}

\medskip
Next, by Random twirling toward the diagonal subspace (\cref{pro:z-twirl}), choosing $T = \Theta(k)$ so that $2^T \ge 2^{11}\cdot 3^{3k}$ ensures, with probability at least $3/4$,
\begin{align}
    \|H_T'\|_F & \le \frac{2}{2^{T/2}}\|H-H_0\|_F \nonumber \\
    & \le \frac{1}{16}\cdot 3^{-k}\cdot \|H_{\mathrm{eff}}\|_F.
\end{align}

\medskip
Applying Stability of eigenvalue-gap proportion (\cref{pro:eig-stability}) to $H_T = H_{\mathrm{eff}} + H_T'$ with
\(p=9^{-k}/4\) and \(q = 3^{-k}/16\) gives
$p=9^{-k}/4$ and $q=3^{-k}/16$ gives
\begin{equation}
    \Lambda\left(H_T,\, \frac12\|H_{\mathrm{eff}}\|_F\right) \ge p-32q^2 = \frac{1}{8}\cdot 9^{-k}.    
\end{equation}

\medskip
Let $\eta=\|H_{\mathrm{eff}}\|_F / 2\ge {\varepsilon}/(2\sqrt{2}\cdot3^{k/2})$ and $d = 9^{-k}/8$.
By sufficient spectral condition (\cref{lem:spectral-condition}), for $t$ uniform in $[0,2/\eta]$ the identity outcome probability satisfies $I_T(t)\le 1 - \Omega(d)$ with constant probability.
Since $2/\eta = \mathcal{O}(3^{k/2}/\varepsilon)$, setting $b = \Theta(3^{k/2}/\varepsilon)$ suffices.

\medskip
Distinguishing $I_T(t)=1$ from $I_T(t)\le1-\Omega(d)$ requires 
$m=\Theta(1/d)=\Theta(9^k)$ Bell samples.
Thus each outer round rejects with probability $\Omega(3^{-k})$, and repeating
$R=\Theta(3^k \log(1/\delta))$ rounds yields overall success probability at least $1-\delta$.

\medskip
Each Bell sample uses a Trotterized simulation of $e^{-itH_T}$ with total evolution time $\mathcal{O}(t)$ and $t\le b=\Theta(3^{k/2}/\varepsilon)$. Hence the total evolution time is
\begin{align}
    \mathcal{O}\left(R \cdot m \cdot b\right) & = \mathcal{O}\left(3^k\log(1/\delta)\cdot 9^k \cdot \frac{3^{k/2}}{\varepsilon}\right) \nonumber \\
    &= \mathcal{O}\left(\frac{c^k\log(1/\delta)}{\varepsilon}\right).
\end{align}

\medskip
Therefore, for constant locality $k=\mathcal{O}(1)$ and constant failure probability~$\delta$, the total evolution time of our algorithm is $\mathcal{O}(1/\varepsilon)$. Moreover, by the matching lower bound of $\Omega(1/\varepsilon)$ established in~\cite{kallaugher2025hamiltonianlocalitytestingtrotterized}, there exist constant-local Hamiltonians such as the simple pair $H=\varepsilon X$ and $H_0=-\varepsilon X$ that require total evolution time at least $\Omega(1/\varepsilon)$ to distinguish in~\cref{problem:intolerant}. Consequently, our result achieves the tight $\Theta(1/\varepsilon)$ bound for intolerant certification in the constant-locality regime.
\end{proof}

\begin{lemma}[(Concrete constants for~\cref{alg:main-algorithm})]\label{lem:constants}
There exist universal numerical constants $C_1,C_2,C_3,C_4>0$ and $c_0\in(0,1)$ such that~\cref{alg:main-algorithm} achieves the guarantees stated in~\cref{thm:main-algorithm}. In particular, one may take
\begin{equation}
    C_1=\frac{16}{3},\qquad
    C_2=17,\qquad
    C_3=4\sqrt{2},\qquad
    C_4 =128,\qquad
    c_0=\frac{1}{64}.
\end{equation}
\end{lemma}

\begin{proof}
The choice $C_1=16/3$ ensures that repeating the outer loop
$R=C_1\cdot 3^{k}\log(1/\delta)$ times amplifies the per-round detection probability 
$(3/4)\cdot(1/(4\cdot 3^{k}))$ to at least $1-\delta$.

\medskip
For the twirling length, \cref{pro:z-twirl} yields
$\|H_T'\|_F \le 2\cdot 2^{-T/2}\|H-H_0\|_F$ with probability at least $3/4$.
To guarantee the perturbative regime required for eigenvalue-gap stability, it suffices to choose $T$ such that
\begin{equation}
    2\cdot 2^{-T/2} \le 
    \frac{1}{16}\cdot 3^{-k}\cdot \frac{1}{\sqrt{2}\,3^{k/2}}.
\end{equation}
This is implied by $2^{T}\ge 2^{11}\cdot 3^{3k}$.
Since $T\ge 11+(3\log 3)k$ and $3\log 3 < 17$, one may take $C_2=17$.

\medskip
For the time range, with $\eta=\tfrac12\|H_{\mathrm{eff}}\|_F\ge\varepsilon/(2\sqrt2\cdot 3^{k/2})$, the choice $b=2/\eta$ ensures $b\le 4\sqrt2\cdot 3^{k/2}/\varepsilon$, so $C_3=4\sqrt2$ suffices.

\medskip
Finally, in the \textsf{ACCEPT} case the ideal identity probability is $1$, whereas in the \textsf{REJECT} case,
\cref{lem:spectral-condition} yields
$I_T(t)\le 1-d/4$ with $d=9^{-k}/8$, i.e., a gap of $1/(32\cdot 9^{k})$.
Choosing $\varepsilon_{\mathrm{Trott}}\le 1/(128\cdot 9^{k})$ and threshold 
$1-c_0/9^{k}$ with $c_0=1/64$ preserves this gap under implementation error.
Thus the algorithm requires $128\cdot 9^{k}$ Bell samples per round, and we may set $C_4=128$.
\end{proof}

\paragraph{Query complexity.}
Each Trotterized implementation of $e^{-itH_T}$ uses  $\mathcal{O}\left(2^{T}\,\sqrt{{t}/{\varepsilon_{\mathrm{Trott}}}}\right)$ oracle queries, as guaranteed by~\cref{thm:trotter}.  
Since the algorithm samples $t$ uniformly from $[0,b]$ with $b=\Theta(3^{k/2}/\varepsilon)$ and performs
$R=\Theta(3^{k}\log(1/\delta))$ rounds, each of which uses $m=\Theta(9^{k})$ Bell samples, multiplying these contributions together yields an overall query count of
\begin{equation}
    \mathcal{O}\left(
        \frac{2^{T}\cdot9^{k}\cdot3^{k/2}\log(1/\delta)}{\varepsilon^{3/2}}
    \right)
    = \mathcal{O}\left(
        \frac{c^k_*\log(1/\delta)}{\varepsilon^{3/2}}
    \right)
\end{equation}
for some universal constant $c_*>1$, where the $k$-dependence from $T$ and $\varepsilon_{\mathrm{Trott}}$ has been absorbed into $c^k_*$.

\addcontentsline{toc}{section}{Acknowledgements}
\section*{Acknowledgments}
We thank Savar D. Sinha and Yu Tong for insightful discussions.

\addcontentsline{toc}{section}{References}
\bibliographystyle{alpha}
\bibliography{citation.bib}

\newpage
\appendix
\section{Deferred proofs}
\subsection{Proof of~\cref{pro:random-basis} (Randomized diagonal basis selection)}\label{app:proof-random-basis}

\begin{proof}
Write the normalized Frobenius norm in terms of the Pauli coefficients:
\begin{equation}
    \|H\|_F^2 \;=\; \sum_{P} \alpha_P^2 \;\ge\; \varepsilon^2.
\end{equation}
For each Pauli string $P$, denote its weight (number of non-identity tensor factors) by $|P|$.
(Assume $H$ is $k$-local, i.e., $\alpha_P=0$ whenever $|P|>k$.)

\medskip
Given random $Q_1,\dots,Q_n\in\{X,Y,Z\}$, the term $\alpha_P P$ appears in $H_Q$ if and only if, for every
qubit $i$ in the support of $P$, the non-identity Pauli of $P$ on that qubit equals $Q_i$.
Since each $Q_i$ is uniform on $\{X,Y,Z\}$, we have
\begin{equation}
    \Pr[\, P\ \text{survives in }H_Q \,] = 3^{-|P|}.
\end{equation}
Define indicator random variables $X_P\in\{0,1\}$ by
\[
    X_P =
    \begin{cases}
    1, & \text{if $P$ survives in $H_Q$},\\
    0, & \text{otherwise},
    \end{cases}
\]
and set $\beta_P \coloneqq \alpha_P^2$. Then
\begin{equation}
    Y \coloneqq \|H_Q\|_F^2 = \sum_{P} \beta_P X_P .
\end{equation}

\medskip
\noindent\textit{First moment.}
\begin{align}
    \mathds{E}[Y]
    &= \sum_{P}\beta_P\,\mathds{E}[X_P]
     = \sum_{P}\beta_P\,3^{-|P|} \nonumber \\
    & \ge 3^{-k}\sum_{P}\beta_P \nonumber \\
    & = 3^{-k}\|H\|_F^2.
\end{align}

\medskip
\noindent\textit{Second moment.}
\begin{align}
    \mathds{E}[Y^2]
    &= \mathds{E}\left[\left(\sum_{P}\beta_P X_P\right)^2\right] \notag\\
    &= \sum_{P}\beta_P^2 \, \mathds{E}[X_P]
       + 2\sum_{P<P'} \beta_P \beta_{P'} \, \mathds{E}[X_P X_{P'}].
\end{align}
For any $P,P'$, the event $\{X_P=X_{P'}=1\}$ can occur only if $P$ and $P'$ agree on every qubit in
$\mathrm{supp}(P)\cap\mathrm{supp}(P')$; otherwise $\mathds{E}[X_PX_{P'}]=0$.
In the consistent case it requires fixing $Q_i$ on all qubits in $\mathrm{supp}(P)\cup\mathrm{supp}(P')$, hence $\mathds{E}[X_PX_{P'}] \;=\; 3^{-|\mathrm{supp}(P)\cup\mathrm{supp}(P')|}$. Therefore,
\begin{align}
\mathds{E}[X_PX_{P'}] & \le 3^{-|\mathrm{supp}(P)\cup\mathrm{supp}(P')|} \nonumber \\
& \le 3^{-(|P|+|P'|)/2},
\end{align}
since $|\mathrm{supp}(P)\cup\mathrm{supp}(P')|\ge (|P|+|P'|)/2$.

\medskip
Therefore
\begin{align}
    \mathds{E}[Y^2]
    &\le \sum_{P}\beta_P^2 \, 3^{-|P|}
      + 2\sum_{P<P'} \beta_P\beta_{P'} \, 3^{-(|P|+|P'|)/2} \notag\\
    &= \left(\sum_{P}\beta_P \, 3^{-|P|/2}\right)^2.
\end{align}
Define $t_P := \beta_P 3^{-|P|/2}$. Then
\begin{equation}
    \mathds{E}[Y^2] \le \left(\sum_{P} t_P\right)^2.
\end{equation}
On the other hand,
\begin{align}
    \mathds{E}[Y]
    & = \sum_{P}\beta_P 3^{-|P|} \nonumber \\
    & = \sum_{P} t_P 3^{-|P|/2} \nonumber \\
    & \ge 3^{-k/2} \sum_{P} t_P,
\end{align}
since $|P|\le k$ implies $3^{-|P|/2}\ge 3^{-k/2}$. Hence
\begin{equation}
    \sum_{P} t_P \le 3^{k/2}\,\mathds{E}[Y]
    \quad\Rightarrow\quad
    \mathds{E}[Y^2] \le 3^k (\mathds{E}[Y])^2.
\end{equation}

\medskip
\noindent\textit{Applying Paley--Zygmund.}
Apply Paley--Zygmund inequality (\cref{lem:p-z}) to $Y$ with $\theta = 1/2$:
\begin{align}
    \Pr\left[ Y \ge \frac12 \, \mathds{E}[Y] \right]
    & \ge \left(1-\frac12\right)^2 \frac{\big(\mathds{E}[Y]\big)^2}{\mathds{E}[Y^2]} \nonumber \\
    & \ge \frac{1}{4}\cdot \frac{1}{3^k}.
\end{align}
On this event,
\begin{align}
    \|H_Q\|_F^2 & = Y \nonumber \\
    & \ge \frac12\,\mathds{E}[Y] \nonumber \\
    & \ge \frac12\cdot 3^{-k}\|H\|_F^2,
\end{align}
and therefore
\begin{align}
    \|H_Q\|_F & \ge \frac{\|H\|_F}{\sqrt{2}\cdot 3^{k/2}} \nonumber \\
    & \ge \frac{\varepsilon}{\sqrt{2}\cdot 3^{k/2}}.
\end{align}
Thus,
\begin{equation}
    \Pr\left[\|H_Q\|_F \ge \frac{\|H\|_F}{\sqrt{2}\cdot 3^{k/2}}\right]
    \ge \frac{1}{4}\cdot \frac{1}{3^k}.
\end{equation}
This completes the proof.
\end{proof}

\subsection{Proof of~\cref{pro:z-twirl} (Random twirling toward the diagonal subspace)}\label{app:proof-z-twirl}
\begin{proof}
We first show that $H_i = H_{\mathrm{eff}} + H_i'$ for all $i$.
Since $S=\bigotimes_{j=1}^n\{I,Q^{(j)}\}$ consists of Paulis that are either $I$ or the fixed single-qubit Pauli $Q^{(j)}$ on each qubit,
all elements of $S$ commute with each other. Hence each random $P_i\in S$ commutes with every $P\in S$, and therefore
\begin{equation}
    H_{\mathrm{eff}} = \frac12 (H_{\mathrm{eff}} + P_i H_{\mathrm{eff}} P_i).
\end{equation}

Assume $H_{i-1}=H_{\mathrm{eff}}+H_{i-1}'$. Then
\begin{align*}
H_i
&= \frac12 (H_{i-1}+P_i H_{i-1} P_i) \nonumber \\
&= \frac12 (H_{\mathrm{eff}}+P_i H_{\mathrm{eff}} P_i)
   + \frac12 (H_{i-1}'+P_i H_{i-1}' P_i) \nonumber \\
& = H_{\mathrm{eff}} + H_i'.
\end{align*}
By induction, $H_i = H_{\mathrm{eff}}+H_i'$ holds for all $i$.

\medskip
Now fix a Pauli string $P\notin S$. Let $B(P) := \{ j\in[n] : P_j\neq I \ \text{and}\ P_j\neq Q^{(j)} \}$. Then $B(P)\neq\emptyset$ (this is exactly $P\notin S$). For each $j\in B(P)$, the single-qubit Paulis $P_j$ and $Q^{(j)}$ anticommute.
A random $P_i\in S$ includes $Q^{(j)}$ on qubit $j$ independently with probability $1/2$, so the overall commutation sign between $P$ and $P_i$
is $(-1)^{\sum_{j\in B(P)} \mathbf{1}[\,P_i \text{ uses }Q^{(j)}\text{ on }j\,]}$, which is equally likely to be $+1$ or $-1$ since $B(P)\neq\emptyset$.
Therefore,
\begin{equation}
    \Pr[\, P \text{ commutes with } P_i \,] = \frac12.    
\end{equation}

\medskip
If $P$ anticommutes with $P_i$, then
\begin{equation}
    \frac12 (P + P_i P P_i) = \frac12(P - P) = 0,    
\end{equation}
so the $P$-coefficient is killed at that step; if $P$ commutes, it survives unchanged. Hence, over $T$ independent steps,
\begin{equation}
    \Pr[\, P \text{ survives in } H_T' \,] = \left(\frac12\right)^T = 2^{-T}.    
\end{equation}

Define indicators $X_P\in\{0,1\}$ for $P\notin S$ by $X_P=1$ iff $P$ survives in $H_T'$, and set $\beta_P:=\alpha_P^2$.
Then
\begin{equation}
    \|H_T'\|_F^2 \;=\; \sum_{P\notin S} \beta_P X_P.    
\end{equation}
Taking expectation gives
\begin{align}
    \mathds{E}[\|H_T'\|_F^2]
& = \sum_{P\notin S} \beta_P\,\mathds{E}[X_P] \nonumber \\
& = 2^{-T}\sum_{P\notin S}\alpha_P^2 \nonumber \\
& = 2^{-T}\,\|H_1'\|_F^2 \nonumber \\
& \le 2^{-T}\,\|H\|_F^2.    
\end{align}

Finally, apply Markov's inequality to the nonnegative random variable $Y:=\|H_T'\|_F^2$:
\begin{equation}
    \Pr\left[ Y \ge 4\,\mathds{E}[Y] \right] \le \frac14,    
\end{equation}
so with probability at least $3/4$,
\begin{align}
    \|H_T'\|_F^2 & \le 4\,\mathds{E}[Y] \nonumber \\
    & \le 4\cdot 2^{-T}\|H\|_F^2,    
\end{align}
i.e.,
\begin{equation}
    \|H_T'\|_F \le \frac{2}{2^{T/2}}\|H\|_F.    
\end{equation}
This proves the claim.
\end{proof}

\end{document}